\documentclass[
    aps,
    pra,
    reprint,            % Mimics the journal's two-column layout
    superscriptaddress, % Essential for associating authors with multiple affiliations
    amsmath,
    amssymb,
    floatfix            % Fixes common float placement errors (figures/tables)
]{revtex4-2}

\usepackage{graphicx}   % Include figure files
\usepackage{dcolumn}    % Align table columns on decimal point
\usepackage{bm}         % Bold math (standard APS recommendation)
\usepackage{mathtools}  % Fixes/extensions to amsmath
\usepackage{amsthm}     % Theorem environments
\usepackage{xcolor}     % For text coloring
\usepackage{braket}
\usepackage{bbm}

\usepackage{float}
\makeatletter
\let\newfloat\newfloat@ltx
\makeatother

\usepackage{algorithm}
\usepackage{algpseudocode}

\usepackage[
    colorlinks=true,
    linkcolor=blue,
    citecolor=blue,
    urlcolor=blue
]{hyperref}

\newtheorem{definition}{Definition}
\newtheorem{theorem}{Theorem}
\newtheorem{lemma}{Lemma}
\newtheorem{proposition}{Proposition}

\newtheorem{remark}{Remark}
\newtheorem{example}{Example}

\begin{document}

\preprint{APS/123-QED}

\title{Restriction-Based Certificate of Bipartite Schmidt Rank in Hypergraph States}

\author{C. Fajardo}
\affiliation{IMDEA Networks, Avda. Mar Mediterr{\'a}neo, 22, 28918 Legan{\'e}s, Spain}
\author{M. Paraschiv}
\affiliation{IMDEA Networks, Avda. Mar Mediterr{\'a}neo, 22, 28918 Legan{\'e}s, Spain}
% \author{Ann Author}
%  \altaffiliation[Also at ]{Physics Department, XYZ University.}%Lines break automatically or can be forced with \\
% \author{Second Author}%
%  \email{Second.Author@institution.edu}
% \affiliation{%
%  Authors' institution and/or address\\
%  This line break forced with \textbackslash\textbackslash
% }%

\begin{abstract}
We investigate bipartite entanglement in qubit hypergraph states across an arbitrary fixed bipartition. 
Using the real equally weighted (REW) representation, the Schmidt rank across the cut can be computed as the real rank of a phase-cleaned cross-cut sign matrix. 
Whereas graph states admit an exact cut-rank rule, because the cross-cut phase is purely bilinear, hypergraph states typically contain higher-degree cross-cut interactions, for which the cut-rank rule fails. 
Our approach certifies entanglement by fixing a single computational-basis assignment on a subset of qubits, thereby selecting a submatrix on an \emph{active slice}. 
When this restriction removes all higher-degree cross-cut residues, the remaining cross-cut phase becomes bilinear up to cut-local terms. 
We call the resulting submatrices \emph{residual-free bilinear cores} and show that they yield an exponential Schmidt-rank lower bound in terms of the $\mathbb{F}_2$-rank of an exposed core matrix. 
We further give a combinatorial sufficient condition, phrased as a disjoint bridge matching, that guarantees the existence of large full-rank cores for broad families of CCZ-type bridge patterns, and we present a search-and-verify procedure that constructs and certifies such cores directly from the hyperedge description.
\end{abstract}

%\keywords{Suggested keywords}%Use showkeys class option if keyword
                              %display desired
\maketitle

\section{Introduction}

Highly structured multipartite entangled states are central to quantum information, both as a means of organizing quantum correlations and as concrete resources for various protocols. Among these, some of the most important are graph states, which provide a framework for measurement-based quantum computation, as well as quantum error correction and distributed protocols \cite{Raussendorf2001OneWay,Raussendorf2003Cluster,Hein2004Multiparty,Hein2006}. In studying these states, one often encounters the problem of establishing the entanglement across a given cut. This is important, for example, for tensor network descriptions and classical simulation strategies, as well as for the entanglement that can be localized between subsystems by measurements on the complement  \cite{VanDenNest2007Sim,Verstraete2004EntanglementCorrelations,Popp2005Localizable}.

A natural generalization of graph states is given by hypergraph states \cite{Qu2013EncodingHypergraphs, Rossi2013Hypergraph}, where multiqubit controlled phase interactions generate entanglement patterns which are not reducible to pairwise edges \cite{Rossi2013Hypergraph}. Hypergraph states are equivalent to real equally weighted (REW) states \cite{Rossi2013Hypergraph}, written in terms of Boolean phase functions, used in this form, for example, in standard oracle-based algorithms \cite{Qu2013Encoding}. Conceptually, hypergraph states also belong to a larger class of locally maximally entanglable states \cite{Kruszynska2009LME,GuHne2014Hypergraph}. In measurement-based models, higher-body phase gates (for example CCZ gates) naturally lead to using hypergraph states as resources \cite{MillerMiyake2016Hierarchy,Gachechiladze2019Depth,Takeuchi2019HypergraphUniversality}. In addition, efficient certification and verification procedures have been developed for them based on local Pauli measurements, which supports their use as practical resources  \cite{Morimae2017Verification,ZhuHayashi2019Verification}, while recent experiments have demonstrated preparation, verification, and processing of complete classes of 4-qubit hypergraph states in photonic platforms \cite{Huang2024Demonstration}. From the perspective of entanglement theory, a substantial body of work has investigated local-unitary classification and symmetries  \cite{GuHne2014Hypergraph,Lyons2015LUSym,Gachechiladze2017Graphical}, multipartite entanglement witnesses and bounds \cite{Ghio2018Witnesses}, purification protocols \cite{Vandre2023Purification}, and typical entanglement properties in random hypergraph ensembles \cite{Zhou2022RandomHypergraph}.  For a broader overview of developments and applications, see also the recent concise review \cite{Salem2025Review}.

Entanglement in hypergraph states has been investigated from several angles. Ref.~\cite{Qu2013EncodingHypergraphs} relate local entropic entanglement measures to purely combinatorial quantities via the Hamming weight of suitable adjacent subhypergraphs, while a much broader analysis of local-unitary equivalence classes, entanglement properties as well as non-classical correlations was carried out in Ref.~\cite{GuHne2014Hypergraph}. Going beyond qubits, Ref.~\cite{Xiong2018QuditHypergraphStates} introduced qudit hypergraph states and studied their response to generalized Pauli operations and measurements, connecting bipartite entanglement criteria to rank properties of reduced coefficient matrices. Finally, Ref.~\cite{Appel2022FFE} proposed a more general phase encoding framework of finite function encoding states, in which qubit hypergraph states are a special case of encoding Boolean functions into phases. For the bipartite case, they introduced a dephased normal form obtained via local diagonal operations and emphasized the role of a core submatrix in classification problems for phase-encoded states.

In this paper, we develop a Schmidt-rank certification method for qubit hypergraph states that combines the standard REW sign-matrix representation with local diagonal dephasing (\textit{phase cleaning}) \cite{Appel2022FFE} and with the bilinear cut-rank rule underlying the graph-state case \cite{Hein2006}. 
In spirit, our bounds are submatrix-rank certificates related to those used for (qudit) hypergraph states in Ref.~\cite{Xiong2018QuditHypergraphStates}, but specialized to single-shot computational-basis restrictions that expose a bilinear cross-cut core while excluding residual higher-degree cross terms. 
Concretely, we fix a single global assignment on a subset of qubits (an \emph{active slice}) and analyze the induced submatrix of the phase-cleaned cross-cut sign matrix. 
When the induced cross-cut phase becomes bilinear up to cut-local terms, the resulting \emph{residual-free bilinear core} certifies an exponential Schmidt-rank lower bound.

Our contributions are as follows: 
\begin{itemize}
    \item We \emph{define} \textit{residual-free bilinear cores}: submatrices exposed by a single global restriction (an \textit{active slice}) on which the remaining cross-cut phase is purely bilinear up to cut-local terms.
    \item We derive a Schmidt-rank certificate: any residual-free bilinear core with exposed matrix $\Gamma_{\text{core}}$ yields the bound $\mathrm{SR}_{A|B} \ \ge\ 2^{\mathrm{rank}_{\mathbb{F}_2}\!\left(\Gamma_{\mathrm{core}}\right)}$.
    \item We give a disjoint bridge matching sufficient condition (covering broad CCZ-bridge patterns) that guarantees the existence of large full-rank cores.
    \item We present an explicit search-and-verify procedure that proposes restrictions, checks residual-freeness, constructs $\Gamma_{\text{core}}$, and outputs a certified lower bound.
\end{itemize}

% In this paper, we formalize a restriction-based Schmidt-rank certification framework for hypergraph states built on the standard REW matrix representation as well as local dephasing ideas \cite{Appel2022FFE, Hein2006} . While submatrix rank arguments have been previously used for entanglement detection in hypergraph states \cite{Xiong2018QuditHypergraphStates}, we provide a general exponential lower bound mechanism by identifying restricted slices where the cross-phase is residual-free bilinear, enabling a cut-rank bound. We then provide a hypergraph combinatorial sufficient condition (disjoint bridge matching) that guarantees that such a slice exists, as well as an algorithmic search and verify procedure.

% Our contributions are as follows: 
% \begin{itemize}
%     \item We introduce the notion of \textit{residual-free bilinear cores} exposed by a single global restriction (which we shall denote as the \textit{active slice}) such that the restricted cross-phase becomes bilinear up to cut-local terms.
%     \item We prove that any residual-free bilinear core with exposed matrix $\Gamma_{\text{core}}$ gives an exponential lower bound $\mathrm{SR}_{A|B} \ \ge\ 2^{\mathrm{rank}_{\mathbb{F}_2}\!\left(\Gamma_{\mathrm{core}}\right)}$.
%     \item We provide a disjoint bridge matching condition (covering broad CCZ-bridge patterns) that guarantees the existence of large full-rank cores.
%     \item We give an explicit procedure that proposes restrictions, checks residual-freeness constructs $\Gamma_{\text{core}}$ and outputs a certified bound.
% \end{itemize}

The paper is organized as follows. Sec.~\ref{sec:prelim} provides preliminary notation and definitions and Sec.~\ref{sec:baseline} reproduces the standard cross-cut Schmidt rank rule for graphs, together with a series of examples identifying cases in which this rule cannot be directly applied (higher-degree edges across the cut). In Sec.~\ref{sec:results} we present a restriction-based method that applies when the cross-cut phase term is not purely bilinear. We provide a Schmidt-rank certificate as well as a sufficient combinatorial condition that ensures the existence of a lower-bound on cross-cut entanglement. This section also contains a proposed algorithm for certificate search and verification. Sec.~\ref{sec:discussion} presents a discussion of the results while Sec.~\ref{sec:conclusion} concludes the paper.

\section{Preliminaries}
\label{sec:prelim}

In this section we introduce the fundamental notation used throughout the paper as well as the required concepts. We review the real equally weighted (REW) description of hypergraph states in terms of Boolean phase polynomials and introduce a bipartition-dependent coefficient matrix, whose singular values are the Schmidt coefficients. We then discuss a process denoted as \textit{phase cleaning}, which removes cut-local phases, residing in only one of the subsystems, and isolates the cross-cut phase data relevant for detecting bipartite entanglement.

Throughout this paper we consider $n$ qubits labeled by the vertex $V=[n]\coloneqq\{1,2,\dots,n\}$. A bit string $x=(x_1,\dots,x_n)\in\{0,1\}^n$ will be described in the computational basis as $\ket{x}\equiv \ket{x_1}\otimes\cdots\otimes\ket{x_n}$. Furthermore, unless otherwise stated, we will work in  $\mathbb{F}_2=\{0,1\}$ with operations such as addition $\oplus$ being implied as modulo 2 (XOR). 

% ---------------------------------------------------------
\subsection{Hypergraph states and Boolean phase polynomials}
\label{subsec:prelim-rew}

We define a hypergraph as the pair $G=(V,E)$ with a vertex set $V=[n]$ and a hyperedge set $E\subseteq\mathcal{P}(V)$. If all hyperedges have degree $2$, one recovers an ordinary graph.

In the context of hypergraph states, each hyperedge $e\subseteq[n]$ represents a diagonal controlled-phase gate $CZ_e$ that contributes a sign to the computational basis depending on whether all qubits in $e$ are in the state $\ket{1}$ \cite{Rossi2013Hypergraph,GuHne2014Hypergraph}.
Thus, for every $x\in\{0,1\}^n$,
\begin{equation}
CZ_e \ket{x}
=
(-1)^{\prod_{i\in e} x_i}\,\ket{x},
\qquad
\prod_{i\in e}x_i \in \{0,1\}.
\label{eq:CZ-action}
\end{equation}
The gates $CZ_e$ are diagonal, thus they mutually commute.

\begin{definition}
\label{def:hypergraph-state}
The \emph{hypergraph state} associated to $G=(V,E)$ is given by \cite{Rossi2013Hypergraph,Qu2013Encoding}
\begin{equation}
\ket{G}
\;\coloneqq\;
\Big(\prod_{e\in E} CZ_e\Big)\ket{+}^{\otimes n},
\label{eq:hypergraph-state}
\end{equation}
where $\ket{+}=(\ket{0}+\ket{1})/\sqrt{2}$.
\end{definition}

In what follows, we will encode the diagonal phases of the hypergraph state by a Boolean function. Any Boolean function $f:\{0,1\}^n\to\{0,1\}$ admits a unique algebraic normal form (ANF) over $\mathbb{F}_2$ \cite{ODonnell2014AOBF} given by
\begin{equation}
f(x)
=
\bigoplus_{S\subseteq[n]} a_S \prod_{i\in S} x_i,
\qquad
a_S\in\{0,1\}.
\label{eq:anf}
\end{equation}
For a hypergraph $G=(V,E)$, each hyperedge contributes the monomial term $\prod_{i\in e} x_i$, yielding the associated phase polynomial
\begin{equation}
f_G(x)
\;=\;
\bigoplus_{e\in E}\;\prod_{i\in e} x_i.
\label{eq:fG}
\end{equation}
Conversely, any Boolean function in an ANF of the form Eq.~\eqref{eq:fG} specifies a hypergraph \cite{Dutta2018Boolean}. Using this, we can rewrite hypergraph states in what is known as the real equally weighted representation (REW) \cite{Rossi2013Hypergraph} where all computational basis amplitudes have a magnitude $2^{-n/2}$ and where the Boolean phase function dictates the sign pattern \cite{Rossi2013Hypergraph,GuHne2014Hypergraph}. In particular, one can write the standard expansion

\begin{equation}
\ket{G} = 2^{-n/2}\sum_{x\in\{0,1\}^n} (-1)^{f_G(x)}\,\ket{x}.
\label{eq:rew}
\end{equation}
This identity follows directly from the definition Eq.~\eqref{eq:hypergraph-state} and is stated, for example, in \cite{Rossi2013Hypergraph,Dutta2018Boolean,ChenWang2014FourQubit}.

% ---------------------------------------------------------
\subsection{Bipartitions, coefficient matrices, and phase cleaning}
\label{subsec:prelim-bipartition}

We will fix a bipartition of our underlying hypergraph and denote it $A|B$, by which we mean two disjoint sets of qubits for which  $A\cup B=[n]$ and $A\cap B=\varnothing$. After fixing an ordering for all qubits in $A$ and $B$, every string of bits $x\in\{0,1\}^n$ can be written uniquely as $x = (u,v)$ with $u\in\{0,1\}^a$ and $v\in\{0,1\}^b$, where $u$ and $v$ represent the bits in $A$ and $B$, respectively.

Any pure state $\ket{\psi}\in \mathcal{H}_A\otimes\mathcal{H}_B$ can be expanded as
\begin{equation}
\ket{\psi} = \sum_{u\in\{0,1\}^a}\sum_{v\in\{0,1\}^b} M_{u,v}\,\ket{u}_A\ket{v}_B,
\label{eq:psi-M}
\end{equation}

where $M\in\mathbb{C}^{2^a\times 2^b}$ is the \emph{coefficient matrix} across the cut. In this form, the subsystem states are given by 
\begin{equation}
\begin{aligned}
\rho_A &\coloneqq \operatorname{Tr}_B(\ket{\psi}\!\bra{\psi}) = M M^\dagger,\\
\rho_B &\coloneqq \operatorname{Tr}_A(\ket{\psi}\!\bra{\psi}) = M^\dagger M.
\end{aligned}
\label{eq:rho-MMdag}
\end{equation}
thus the singular values of $M$ represent the Schmidt coefficients of $\ket{\psi}$, and in particular the Schmidt rank across the cut $A|B$ is equal to $\mathrm{rank}(M)$. For the hypergraph state Eq.~\eqref{eq:rew}, the coefficient matrix takes the form
\begin{equation}
M_{u,v}
=
2^{-n/2}(-1)^{f_G(u,v)}\in\mathbb{R}.
\label{eq:M-hypergraph}
\end{equation}
Since $M$ is real for hypergraph states, we may replace $\dagger$ by $\mathsf{T}$ and use the transpose instead.

The ANF of $f_G$ decomposes relative to the cut $A|B$ into three components 
\begin{equation}
f_G(u,v)
=
f_A(u)\ \oplus\ f_B(v)\ \oplus\ f_{AB}(u,v),
\label{eq:phase-decomp}
\end{equation}
where $f_A$ contains only the monomial terms supported on variables in $A$, $f_B$ represents the terms from variables in $B$ while $f_{AB}$ contains all remaining monomials that consist of at least one variable from each side of the bipartition. Any constant term can either be integrated into this form or ignored as part of a global phase. 

In the rest of the paper, we are interested solely in the cross-cut terms and will thus remove the $f_G$ components that belong to only one bipartition. This can be done by defining two diagonal matrices on $\mathcal{H}_A$ and $\mathcal{H}_B$ as
\begin{equation}
(D_A)_{u,u} \coloneqq (-1)^{f_A(u)},
\qquad
(D_B)_{v,v} \coloneqq (-1)^{f_B(v)}.
\label{eq:DA-DB}
\end{equation}
and observe that the coefficient matrix of Eq.~\eqref{eq:M-hypergraph} factorizes as
\begin{equation}
M
=
2^{-n/2}\, D_A\, R\, D_B,
\qquad
R_{u,v}\coloneqq (-1)^{f_{AB}(u,v)}\in\{\pm1\}.
\label{eq:M-factor}
\end{equation}
As the diagonal matrices $D_A$ and $D_B$ have no effect on the singular values of $M$ or on its rank, when computing cross-cut entanglement, one must only focus on the \emph{phase-cleaned sign matrix} $R$ ~\cite{Appel2022FFE}.

\begin{lemma}
\label{lem:phase-cleaning}
Let $U_A$ and $U_B$ be unitaries acting on $A$ and $B$, respectively. Then the coefficient matrices of $\ket{\psi}$ and $(U_A\otimes U_B)\ket{\psi}$ are related by $M\mapsto U_A M U_B^{\mathsf{T}}$ (in the computational bases), and in particular $M$ and $U_A M U_B^{\mathsf{T}}$ share the same singular values and the same rank.
For hypergraph states, the diagonal matrices $D_A$ and $D_B$ in Eq.~\eqref{eq:M-factor} are also unitaries, hence all bipartite entanglement properties across $A|B$ depend exclusively on $R$ (thus on the cross-cut polynomial $f_{AB}$).
\end{lemma}

\begin{proof}
Left and right multiplication by unitaries does not affect singular values: if $M=U\Sigma V^\dagger$ is a singular value decomposition (SVD), then $(U_A M U_B^{\mathsf{T}})=(U_A U)\Sigma ( (U_B^*) V)^\dagger$ is another SVD with the same diagonal matrix $\Sigma$. Rank is also preserved because unitaries are invertible. Applying this to Eq.~\eqref{eq:M-factor} shows that the Schmidt spectrum (hence any bipartite entanglement measure for a pure state) is unchanged by multiplying rows and columns of $R$ by phases, and therefore depends only on $R$.
\end{proof}
An important result of Lemma~\ref{lem:phase-cleaning} is that the Schmidt rank across the cut $A|B$ can be computed only in terms of the matrix $R$
\begin{equation}
\mathrm{SR}_{A|B} = \mathrm{rank}(M) = \mathrm{rank}(R).
\label{eq:SR-rankR-prelim}
\end{equation}

As the global system is described by the pure bipartite state Eq.~\eqref{eq:psi-M}, the degree to which one of the subsystem states, either $\rho_A$ or $\rho_B$, is mixed relates directly to the entanglement across the cut $A|B$. In the next section we will focus on one of the subsystem states $\rho_A$ (as both $\rho_A$ and $\rho_B$ share the same nonzero eigenvalue spectrum in this case), and use the amount of correlation between the rows of the corresponding matrix to determine the degree to which the state is mixed. We will see that for the simplest case in which the bipartition $A|B$ is crossed only by 2-edges (regular graph edges), we recover a well-known result for graph states. The same approach does not work, however, if higher degree edges cross the cut. We then propose a method by which we can provide a lower limit in this case. This represents the main result of our paper.

\section{Bilinear Slice and the Graph-State Cut-Rank Rule}
\label{sec:baseline}

The REW form introduced in the previous section offers a clear separation of bipartite entanglement across the cut $A|B$, through the phase-cleaned sign matrix $R_{u,v}=(-1)^{f_{AB}(u,v)}$. A first approximation would be to retain only the \emph{degree-$2$} cross-cut part of $f_{AB}$, which can be used to define a binary matrix $\Gamma_{AB}$, encoding the cross-cut edges modulo 2. In the case of graph states, where the cross-cut phase is exactly bilinear, the matrix $\Gamma_{AB}$ completely determines the Schmidt rank across $A|B$. For more general hypergraph states, however, the same bilinear slice can either be silent (in the absence of bilinear terms) or misleading. 

\subsection{Bilinear slice and cut-incidence matrix $\Gamma_{AB}$}
\label{subsec:bilinear-slice}

Writing the cross-cut phase polynomial component $f_{AB}$ in its ANF form, and keeping only monomials of total degree $2$ which contain one variable from $A$ and one from $B$, we can define the bilinear slice
\begin{equation}
f^{(2)}_{AB}(u,v)
\;=\;
\bigoplus_{i\in A}\ \bigoplus_{j\in B}\ \gamma_{ij}\, u_i v_j.
\qquad
\gamma_{ij}\in\{0,1\},
\label{eq:bilinear-slice}
\end{equation}
The coefficients $\gamma_{ij}$ then define an associated cut-incidence matrix
\begin{equation}
\Gamma_{AB}\coloneqq[\gamma_{ij}] \in \mathbb{F}_2^{a\times b},
\label{eq:GammaAB}
\end{equation}
behaving as an effective cross-adjacency matrix for the cut $A|B$. Equivalently, after identifying $u\in\mathbb{F}_2^a$ and $v\in\mathbb{F}_2^b$ as column vectors, one may also write
\begin{equation}
f^{(2)}_{AB}(u,v) \;=\; u^\top \Gamma_{AB}\, v \quad (\mathrm{mod}\ 2).
\label{eq:bilinear-form}
\end{equation}
Intuitively, $\Gamma_{AB}$ records which \emph{cross} $2$-hyperedges bridge the cut, with coefficients taken modulo $2$.

\subsection{Bilinear case: $\mathrm{SR}_{A|B}=2^{\mathrm{rank}_{\mathbb{F}_2}(\Gamma_{AB})}$ and limitations}
\label{subsec:bilinear-cutrank}

We may now state and prove the bipartite entanglement rank rule for graph states, expressed using the introduced notation. It is convenient to first define the \emph{log-Schmidt-rank}
\begin{equation}
\begin{aligned}
E_{\mathrm{SR}}(A|B)
&\;\coloneqq\;\log_2 \mathrm{SR}_{A|B}\\
&\;=\;\log_2 \mathrm{rank}(M)\\
&\;=\;\log_2 \mathrm{rank}(R).
\end{aligned}
\label{eq:ESR-def}
\end{equation}
where the last equality uses phase cleaning (Sec.~\ref{subsec:prelim-bipartition}) and Eq.~\eqref{eq:SR-rankR-prelim}.

\begin{theorem}
\label{thm:bilinear-cutrank}
Let $A|B$ be a bipartition with $|A|=a$ and $|B|=b$, and let $\Gamma_{AB}\in\mathbb{F}_2^{a\times b}$ be the bilinear-slice matrix defined in Eq.~\eqref{eq:GammaAB}.
Assume also that the \emph{entire} cross-cut phase is purely bilinear, meaning it can be written as 
\begin{equation}
\begin{aligned}
f_{AB}(u,v) &\;=\; f^{(2)}_{AB}(u,v)\\
&\;=\; u^\top \Gamma_{AB} v \qquad (\mathrm{mod}\ 2).
\end{aligned}
\label{eq:pure-bilinear-assumption}
\end{equation}
and let $k\coloneqq \mathrm{rank}_{\mathbb{F}_2}(\Gamma_{AB})$. Then we have that
\begin{equation}
E_{\mathrm{SR}}(A|B)=k,
\qquad
\mathrm{SR}_{A|B}=2^k.
\label{eq:SR-cutrank}
\end{equation}
In particular, for graph states this reproduces the standard Schmidt-rank rule \cite[Prop.~10]{Hein2006}. For completeness, we also provide a sketch proof in Appendix~\ref{app:th1-proof}.
\end{theorem}

In the case of graph states, $f_G$ is quadratic and the $f_{AB}$ component is bilinear. The cut-incidence matrix $\Gamma_{AB}$ coincides in this case with the $A\times B$ adjacency block (the $|A| \times |B|$ sized block of the graph adjacency matrix) where entry $\gamma_{i,j}$ is 1 if and only if the corresponding $(i,j)$ edge crosses the cut. The $\mathrm{rank}_{\mathbb{F}_2}(\Gamma_{AB})$ is also precisely the \emph{cut-rank} of the graph across the cut.

 When all cross-cut monomials have degree $2$ (equivalently, $f_{AB}=f^{(2)}_{AB}$), Theorem~\ref{thm:bilinear-cutrank} is exact and provides a complete Schmidt-rank characterization across $A|B$ \cite{Hein2006}. If however $f_{AB}$ contains only higher-degree cross monomials (for example terms such as $u_i v_j v_k$), then $\Gamma_{AB}=0$ and the bilinear slice provides no information, even though $\mathrm{SR}_{A|B}$ can be large (see ~Example~\ref{ex:pure-3edge-vanishing-bilinear-slice}). Finally, if higher-degree cross terms are present together with terms of degree $2$, the rank of the real sign matrix  $R_{u,v}=(-1)^{f_{AB}(u,v)}$ can show a substantial deviation from the bilinear prediction $2^{\mathrm{rank}_{\mathbb{F}_2}(\Gamma_{AB})}$. For example, adding cubic (or higher) cross terms can \emph{increase or decrease} $\mathrm{rank}(R)$ depending on interference between phases. Appendix~\ref{app:rank3-counterexample} gives an explicit $4\times 4$ instance with $\mathrm{rank}_{\mathbb{F}_2}(\Gamma_{AB})=2$ but $\mathrm{SR}_{A|B}=3<4$.

\begin{example}[Two independent cross-$CZ$ bridges]
\label{ex:two-bridges-bilinear-slice}
Let $A=\{1,2\}$ and $B=\{3,4\}$, and suppose the only cross-cut monomials are
\[
f^{(2)}_{AB}(u,v)=u_1 v_3 \oplus u_2 v_4.
\]
Then
\[
\Gamma_{AB}=
\begin{pmatrix}
1 & 0\\
0 & 1
\end{pmatrix},
\qquad
\mathrm{rank}_{\mathbb{F}_2}(\Gamma_{AB})=2.
\]
If no higher-degree cross monomials are present (such that $f_{AB}=f^{(2)}_{AB}$), the cut-rank rule of Sec.~\ref{subsec:bilinear-cutrank} yields $\mathrm{SR}_{A|B}=4$.
\end{example}

\begin{example}[Pure cross $3$-edge: vanishing bilinear slice]
\label{ex:pure-3edge-vanishing-bilinear-slice}
Let $A=\{1\}$ and $B=\{2,3\}$, and take a single cross hyperedge $\{1,2,3\}$.
Then $f_{AB}(u,v)=u_1 v_2 v_3$ has no degree-$2$ cross terms, so $f^{(2)}_{AB}\equiv 0$ and $\Gamma_{AB}=0$.
Nevertheless the state is generally entangled across $A|B$; in Section~\ref{sec:results} we show how appropriate restrictions (interpretable as postselected $Z$-measurement outcomes  \cite{Gachechiladze2017Graphical,Xiong2018QuditHypergraphStates}) can expose bilinear structure on an active slice even when $\Gamma_{AB}=0$ globally.
\end{example}

These limitations indicate that $\Gamma_{AB}$ should be seen as a baseline feature extraction procedure, rather than as a general cross-cut entanglement estimator for hypergraph states. Theorem~\ref{thm:bilinear-cutrank} can however prove itself useful once more if one can identify a submatrix of $R$ (by for example restricting some qubits to fixed Z-basis values) on which the remaining cross-cut phase is purely bilinear (up to cut-local terms). In this case, the same arguments can be used to derive an exponential rank lower bound. In Section~\ref{sec:results} we provide combinatorial conditions as well as an algorithm for constructing such \textit{residual-free bilinear cores} of $R$ using a single global restriction.

\section{Restriction-Based Residual-Free Bilinear Cores}
\label{sec:results}

In Section~\ref{sec:baseline} we saw that when the cross-cut phase term is purely bilinear, the Schmidt rank across the bipartition $A|B$ is given by the $\mathbb{F}_2$-rank of the cut-incidence matrix $\Gamma_{AB}$. For more general hypergraph states however, higher-degree cross-terms may dominate and the global bilinear slice may vanish, even if the state still has a large Schmidt rank. In this section we present a restriction-based method that recovers the bilinear structure on an \textit{active slice} of variables, while ensuring that no higher-degree terms are present on that slice. This provides a rigorous exponential lower bound on the Schmidt rank, based on a submatrix of $R$. Our use of phase cleaning and diagonal normal forms is motivated by the dephased normal form viewpoint of Ref.~\cite{Appel2022FFE}.

\subsection{Residual-free bilinear cores and a Schmidt-rank certificate}
\label{subsec:cores-certificate}

We have previously seen how the Schmidt rank across a cut $A|B$ is given by the rank of the phase-cleaned coefficient matrix $R$, as per Eq.~\eqref{eq:SR-rankR-prelim}. Consequently, any entanglement certificate based on a large-rank submatrix of $R$ yields a Schmidt-rank lower bound for the original state. In this section we will use this approach to address the situation in which the edges across the cut $A|B$ have degree higher than $2$. Before providing the certificates however, it is important to define the notion of \textit{active slice} as well as that of a \textit{Boolean derivative} and its effects. 

\paragraph*{Restrictions and active slices.}
Let $I\subseteq A$ and $J\subseteq B$ be sets of \emph{active} indices, and let $\beta\in\{0,1\}^{A\setminus I}$ and $\alpha\in\{0,1\}^{B\setminus J}$ be fixed assignments of the \emph{non-active} variables.
Write $u_I\in\{0,1\}^{|I|}$ and $v_J\in\{0,1\}^{|J|}$ for the remaining free bits.
The corresponding restriction of the cross-cut phase polynomial is
\begin{equation}
\widetilde f^{(I,J;\beta,\alpha)}_{\mathrm{core}}(u_I,v_J)
\;\coloneqq\;
f_{AB}\big(u_I,\ u_{A\setminus I}=\beta;\ v_J,\ v_{B\setminus J}=\alpha\big).
\label{eq:core-restriction-results}
\end{equation}
From the matrix point of view, fixing $(u_{A\setminus I},v_{B\setminus J})=(\beta,\alpha)$ selects a $2^{|I|}\times 2^{|J|}$ submatrix of $R$, indexed by $(u_I,v_J)$.
The guiding principle is to choose $(I,J;\beta,\alpha)$ so that $\widetilde f_{\mathrm{core}}$ is bilinear in $(u_I,v_J)$ up to cut-local terms.

\paragraph*{Boolean derivatives as a constructive heuristic.}
The certificate is formulated in terms of imposed restrictions and remaining monomials, however it is helpful to recall two elementary facts that aid in the search for good restrictions. For a Boolean polynomial $f$, written in its ANF, the Boolean derivative with respect to a variable $x_i$ behaves as a discrete difference
\begin{equation}
\Delta_{x_i} f(x) \;\coloneqq\; f(x)\oplus f(x\oplus e_i),
\label{eq:boolean-derivative}
\end{equation}
where $e_i$ flips the $i$th bit.
If $m(x)=\prod_{t\in T} x_t$ is a monomial, then
\begin{equation}
\Delta_{x_i} m(x)
=
\begin{cases}
\prod_{t\in T\setminus\{i\}} x_t, & i\in T,\\
0, & i\notin T,
\end{cases}
\label{eq:derivative-peels-factor}
\end{equation}
so a derivative with respect to an $A$-variable can eliminate that factor from a cross monomial. A second important fact is that any pure-$B$ monomial $\prod_{j\in T} v_j$ can be linearized into a single variable by setting all but one of its variables to $1$ and setting all other $B$-variables outside $T$ to $0$; concretely, fixing an arbitrary $j^\star\in T$,
\begin{equation}
\begin{aligned}
v_j &\leftarrow 1 && \text{for } j\in T\setminus\{j^\star\},\\
v_j &\leftarrow 0 && \text{for } j\in B\setminus T,\\
\prod_{j\in T} v_j &\mapsto v_{j^\star}. &&
\end{aligned}
\label{eq:linearize-monomial}
\end{equation}
These observations suggest that cross monomials of the form $u_i\prod_{j\in T}v_j$ can be turned into bilinear bridges $u_i v_{j^\star}$ by a suitable restriction on $B$.
The main obstacle is that a rank certificate requires a \emph{single global restriction} that simultaneously linearizes many such monomials and also eliminates any higher-degree cross residuals on the active variables.

\begin{definition}
\label{def:residualfree-core-results}
Let $I=\{i_1,\dots,i_r\}\subseteq A$ and $J=\{j_1,\dots,j_p\}\subseteq B$ be sets of distinct indices, with fixed orderings, and define
\[
\begin{aligned}
u_I &= (u_{i_1},\dots,u_{i_r})\in\{0,1\}^r,\\
v_J &= (v_{j_1},\dots,v_{j_p})\in\{0,1\}^p.
\end{aligned}
\]
For restrictions $(\beta,\alpha)\in\{0,1\}^{A\setminus I}\times \{0,1\}^{B\setminus J}$, define the restricted cross phase $\widetilde f^{(I,J;\beta,\alpha)}_{\mathrm{core}}$ by Eq.~\eqref{eq:core-restriction-results}.
We say that $(I,J;\beta,\alpha)$ exposes a \emph{residual-free bilinear core} if there exist a matrix $\Gamma_{\mathrm{core}}\in\mathbb{F}_2^{r\times p}$ and Boolean polynomials $\phi:\{0,1\}^r\to\{0,1\}$ and $\psi:\{0,1\}^p\to\{0,1\}$ such that
\begin{equation}
\widetilde f^{(I,J;\beta,\alpha)}_{\mathrm{core}}(u_I,v_J)
=
u_I^\top \Gamma_{\mathrm{core}}\, v_J
\ \oplus\ \phi(u_I)\ \oplus\ \psi(v_J).
\label{eq:core-decomposition-results}
\end{equation}

Equivalently, after applying the restriction $(u_{A\setminus I}=\beta,\ v_{B\setminus J}=\alpha)$, no cross monomial of total degree $\ge 3$ survives on the active variables $I\cup J$ \cite{Hein2004Multiparty, Rossi2013Hypergraph}.
\end{definition}
The decomposition in Eq.~\eqref{eq:core-decomposition-results} captures the situation in which the induced sign matrix differs from a purely bilinear character matrix only by diagonal row/column phases (such as the dephased normal form viewpoint of Ref.~\cite{Appel2022FFE}). In Eq.~\eqref{eq:core-decomposition-results} the restricted phase is decomposed into the bilinear core $\Gamma_{\mathrm{core}} \in \mathbb{F}_2^{r \times p}$, which encodes the surviving bilinear cross terms, and two functions $\phi_{u_I}$ and $\psi_{v_J}$ which are generic functions local to the A or B partition respectively. The important point is that no cross-monomial of degree greater than $2$ remains in the active variables. This is what we mean by the \textit{residual-free condition}. Concretely, there are four types of terms that can form the restricted cross phase: \textit{bilinear cross-terms} such as $u_{i_k}v_{j_l}$ for $i_k \in I$ and $j_l \in J$, which enter the $\Gamma_{\mathrm{core}}$ matrix; \textit{Local-I}, representing any monomial in $u_I$ variables only, which are absorbed into $\phi_{u_I}$; \textit{Local-J}, similarly representing monomials in $v_J$ only, included into $\psi_{v_J}$; and finally \textit{higher-degree cross residuals} such as $u_{i_k}u_{i_l}v_{j_m}$, which must be absent from the above expression. If the last category of terms is missing from Eq.~\eqref{eq:core-decomposition-results}, the core is \textit{residual-free} and the decomposition holds. Being able to write $\widetilde f^{(I,J;\beta,\alpha)}_{\mathrm{core}}(u_I,v_J)$ in the form of Eq.~\eqref{eq:core-decomposition-results} is equivalent to not having any monomial of degree greater than $2$ across the cut.

\begin{remark}
\label{rem:beta-alpha-explicit}
It may be tempting to set $\beta=\mathbf{0}$ so that any cross monomial touching an $A$-variable outside $I$ is removed automatically. We choose to keep $(\beta,\alpha)$ explicit because operationally and algorithmically there is no need to hard-code this choice: any single restriction that yields Eq.~\eqref{eq:core-decomposition-results} provides a valid certificate.
\end{remark}

\begin{theorem}
\label{thm:core-certificate-results}
Assume that there exist sets $I\subseteq A$ and $J\subseteq B$ with $|I|=|J|=r$ and a single restriction pair $(\beta,\alpha)\in\{0,1\}^{A\setminus I}\times \{0,1\}^{B\setminus J}$ such that $(I,J;\beta,\alpha)$ exposes a residual-free bilinear core in the sense of Definition~\ref{def:residualfree-core-results}, with core matrix $\Gamma_{\mathrm{core}}$ in Eq.~\eqref{eq:core-decomposition-results}.
Let
\begin{equation}
t \;\coloneqq\; \mathrm{rank}_{\mathbb{F}_2}(\Gamma_{\mathrm{core}}).
\label{eq:t-core-rank}
\end{equation}
Then $E_{\mathrm{SR}}(A|B)\ \ge\ t$ or, equivalently $\mathrm{SR}_{A|B}\ \ge\ 2^t$.
In particular, if $\Gamma_{\mathrm{core}}$ has full $\mathbb{F}_2$-rank $r$, then $\mathrm{SR}_{A|B}\ge 2^r$.
\end{theorem}

\begin{proof}
We will fix $(I,J;\beta,\alpha)$ as per the hypothesis and consider the $2^r\times 2^r$ submatrix $R'$ of $R$ obtained by restricting to rows with $u_{A\setminus I}=\beta$ and columns with $v_{B\setminus J}=\alpha$, while allowing $(u_I,v_J)$ to range freely.
Thus, by construction,
\[
R'_{u_I,v_J} = (-1)^{\widetilde f^{(I,J;\beta,\alpha)}_{\mathrm{core}}(u_I,v_J)}.
\]
Using \eqref{eq:core-decomposition-results}, we can factor
\[
R'_{u_I,v_J}
=
(-1)^{\phi(u_I)}\,(-1)^{u_I^\top\Gamma_{\mathrm{core}} v_J}\,(-1)^{\psi(v_J)}.
\]
Multiplying rows by $(-1)^{\phi(u_I)}$ and columns by $(-1)^{\psi(v_J)}$ does not change matrix rank, so $\mathrm{rank}(R')$ equals the rank of the purely bilinear sign matrix
\[
R''_{u_I,v_J} \;=\; (-1)^{u_I^\top\Gamma_{\mathrm{core}} v_J}.
\]
By the bilinear cut-rank rule (Theorem~\ref{thm:bilinear-cutrank}) applied to the $r$ active bits, $\mathrm{rank}(R'')=2^{\mathrm{rank}_{\mathbb{F}_2}(\Gamma_{\mathrm{core}})}=2^t$, hence $\mathrm{rank}(R')=2^t$.
Since $R'$ is a submatrix of $R$, $\mathrm{rank}(R)\ge \mathrm{rank}(R')=2^t$.
Finally, $\mathrm{SR}_{A|B}=\mathrm{rank}(R)$ by phase cleaning, so $\mathrm{SR}_{A|B}\ge 2^t$, which is our desired result.
\end{proof}
This should also be compared with the submatrix-rank entanglement witness method of Ref.~\cite{Xiong2018QuditHypergraphStates}, providing a similar result.
\begin{remark}
\label{rem:scope-core-results}
Theorem~\ref{thm:core-certificate-results} lower-bounds the Schmidt rank of the \emph{original} hypergraph state across $A|B$ using a high-rank submatrix of $R$. It is a Schmidt-rank certificate and does not, by itself, imply that the entropy of entanglement across $A|B$ equals $t$ or that the Schmidt spectrum is flat.
\end{remark}

\subsection{Disjoint bridge matching: a sufficient combinatorial condition}
\label{subsec:bridge-matching}

In this and the following subsection we present the main results of the paper. Theorem~\ref{thm:core-certificate-results} ensures that, given a residual-free bilinear core, the lower bound on cross-cut entanglement can be specified. In what follows we will provide a sufficient condition, stated in terms of the cross-hyperedge pattern, which guarantees the existence of a large core provided a single restriction. The standard example is a set of \textit{bridges} from one vertex in $A$ to a subset of vertices in $B$, as resulting from CCZ-type hyperedges across the cut.

\begin{proposition}
\label{prop:bridge-certificate-results}
Suppose there exist $r$ cross hyperedges of the form
\begin{equation}
h_k \;=\; \{i_k\}\cup T_k,
\qquad
i_k\in A,\quad T_k\subseteq B,\quad |T_k|\ge 1,
\label{eq:bridge-hk}
\end{equation}
such that the following three properties hold simultaneously.
First, the $B$-supports $T_1,\dots,T_r$ are pairwise disjoint.
Second, there is no other cross hyperedge whose vertex set intersects two distinct blocks $\{i_k\}\cup T_k$ and $\{i_{k'}\}\cup T_{k'}$ with $k\neq k'$.
Third, within each block $\{i_k\}\cup T_k$, the only cross hyperedge supported entirely inside that block (modulo $2$ cancellation) is $h_k$ itself.
Pick an arbitrary representative $j_k\in T_k$ for each $k$, define active sets $I=\{i_1,\dots,i_r\}\subseteq A$ and $J=\{j_1,\dots,j_r\}\subseteq B$, and choose the restriction
\begin{equation}
\label{eq:alpha-bridge-results}
\begin{split}
u_{A\setminus I} &\leftarrow 0, \\
v_j &\leftarrow 1 \quad \text{for } j\in \bigcup_{k=1}^r (T_k\setminus\{j_k\}), \\
v_j &\leftarrow 0 \quad \text{for } j\in B\setminus \bigcup_{k=1}^r T_k.
\end{split}
\end{equation}
Then $(I,J;\beta=\mathbf{0},\alpha)$ exposes a residual-free bilinear core with $\Gamma_{\mathrm{core}}=I_r$ (the $r\times r$ identity), and consequently
\begin{equation}
\mathrm{SR}_{A|B}\ \ge\ 2^r.
\label{eq:bridge-bound-results}
\end{equation}
\end{proposition}
\begin{proof}
Each bridge hyperedge $h_k=\{i_k\}\cup T_k$ contributes the cross monomial $u_{i_k}\prod_{j\in T_k} v_j$ to $f_{AB}$.
Under the restriction \eqref{eq:alpha-bridge-results}, all variables in $T_k\setminus\{j_k\}$ are set to $1$, and so this monomial reduces to $u_{i_k}v_{j_k}$.
Because the sets $T_k$ are disjoint, these reduced monomials involve distinct $u$-variables and distinct $v$-variables.

The restriction $u_{A\setminus I}\leftarrow 0$ eliminates every cross monomial that contains any $A$-variable outside $I$.
The second property in the statement prevents any other cross hyperedge from producing (after restriction) a surviving cross monomial that couples variables belonging to two different blocks; in particular, it forbids terms that would generate couplings $u_{i_k}v_{j_{k'}}$ with $k\neq k'$ or higher-degree cross residuals spanning multiple blocks.
The third property prevents additional cross residuals supported entirely within a single block $\{i_k\}\cup T_k$ from surviving and spoiling residual-freeness on the active set $I\cup J$.
Therefore, on the active variables $(u_I,v_J)$ the surviving cross phase equals $\bigoplus_{k=1}^r u_{i_k}v_{j_k}$ up to local terms, which is precisely Eq.~\eqref{eq:core-decomposition-results} with $\Gamma_{\mathrm{core}}=I_r$.
The bound Eq.~\eqref{eq:bridge-bound-results} then follows from Theorem~\ref{thm:core-certificate-results}.
\end{proof}

\begin{remark}
\label{rem:why-residualfree}
Eq.~\eqref{eq:linearize-monomial} shows that \emph{individual} high-degree monomials can be linearized by a suitable restriction, however a Schmidt-rank certificate requires a \emph{single global restriction} that defines a $2^r\times 2^r$ submatrix of $R$.
Moreover, allowing higher-degree cross monomials to survive on the active variables can invalidate bilinear rank predictions: the real rank of the induced sign matrix can drop below the $2^{\mathbb{F}_2\text{-rank}}$ value due to interference, as illustrated by the counterexample in Appendix~\ref{app:rank3-counterexample}.
Residual-freeness in Definition~\ref{def:residualfree-core-results} is designed precisely to exclude this problem.
\end{remark}

\subsection{Algorithm}
\label{subsec:algorithm-examples}

\begin{algorithm}[t]
\caption{\textsc{SearchAndVerifyCoreCertificate}}
\label{alg:search-verify-core-compact}
\footnotesize
\begin{algorithmic}[1]
\Require Hyperedges $E$ on $A\cup B$; bipartition $A|B$; target $r$ (or maximum $r_{\max}$)
\Ensure Certificate $(I,J;\beta,\alpha)$ with $\Gamma_{\mathrm{core}}$ and $t=\mathrm{rank}_{\mathbb{F}_2}(\Gamma_{\mathrm{core}})$, or \textsc{Fail}
\State $E_{AB}\gets \textsc{CrossEdges}(E,A,B)$
\State $\text{Blocks}\gets \textsc{BridgeBlocks}(E_{AB})$
    \Comment{$(i,T)$ from $e$ with $e\cap A=\{i\}$}
\For{$r \gets \textsc{SizesToTry}(r,r_{\max})$}
  \State $\{(i_k,T_k)\}_{k=1}^r \gets \textsc{SelectDisjointBlocks}(\text{Blocks},r)$
  \If{\textsc{SelectDisjointBlocks} fails} \State \textbf{continue} \EndIf
  \State $(I,J;\beta,\alpha)\gets \textsc{CanonicalRestriction}\bigl(\{(i_k,T_k)\}_{k=1}^r\bigr)$
  \State $\mathcal{S}\gets \textsc{OddReducedSupports}(E_{AB},I,J;\beta,\alpha)$
    \Comment{reduce $+$ mod-2 cancel}
  \If{\textbf{not} $\textsc{ResidualFree}(\mathcal{S},I,J)$} \State \textbf{continue} \EndIf
  \State $\Gamma_{\mathrm{core}}\gets \textsc{CoreMatrix}(\mathcal{S},I,J)$
    \Comment{$\Gamma_{k\ell}\!=\!1\Leftrightarrow \{i_k,j_\ell\}\!\in\!\mathcal{S}$}
  \State $t\gets \textsc{RankF2}(\Gamma_{\mathrm{core}})$
  \State \Return $(I,J;\beta,\alpha),\,\Gamma_{\mathrm{core}},\,t$
\EndFor
\State \Return \textnormal{\textsc{Fail}}
\end{algorithmic}
\end{algorithm}

Proposition~\ref{prop:bridge-certificate-results} suggests a direct search and verify approach can be designed. The initial search phase is tasked with proposing candidate active sets $I\subseteq A$, $J\subseteq B$ and a restriction $(\beta,\alpha)$, aimed at linearizing many cross-hyperedges into bilinear bridges, while the verification phase checks for residual freeness and extracts $\Gamma_{\mathrm{core}}$. Theorem~\ref{thm:core-certificate-results} then guarantees that any candidate which passes verification will yield a mathematically correct Schmidt-rank lower bound.

\paragraph{Inputs and outputs.}
The input is a hypergraph specified by its list of hyperedges $E$, together with the bipartition $A|B$ and either a target size $r$ or a maximal size that the algorithm can search for.
The output is either a valid certificate $(I,J;\beta,\alpha)$ with the associated core matrix $\Gamma_{\mathrm{core}}$ and $t=\mathrm{rank}_{\mathbb{F}_2}(\Gamma_{\mathrm{core}})$, which implies $\mathrm{SR}_{A|B}\ge 2^t$, or a declaration that the current search heuristic failed to find a certificate of the requested size.

\paragraph{Search phase: proposing candidate bridge blocks.}
Given the hyperedges $E$ and bipartition $A|B$, we first extract the cross hyperedges
\begin{equation}
E_{AB} \;\coloneqq\; \{e\in E:\ e\cap A\neq\emptyset\ \text{and}\ e\cap B\neq\emptyset\}.
\label{eq:EAB-def}
\end{equation}
Of particular interest are cross hyperedges that touch exactly one vertex in $A$; for each such hyperedge $e$ with $e\cap A=\{i\}$ and $e\cap B=T$, we record a \textit{bridge block} $(i,T)$ corresponding to the cross monomial $u_i\prod_{j\in T} v_j$.
Motivated by \eqref{eq:derivative-peels-factor}--\eqref{eq:linearize-monomial}, we then seek a collection of $r$ blocks $(i_k,T_k)$ such that the $A$-indices $i_k$ are distinct and the $B$-supports $T_k$ are disjoint.
This is a set-packing problem; in structured instances the disjointness patterns are often clear, while in denser instances a greedy heuristic that is biased towards a small $|T|$ typically reduces conflicts. It is important to note that set packing problem is known to be NP-complete, as well as very difficult to solve via optimization methods \cite{Halldorsson1998}. This must be taken into account when selecting very large $B$ sets. 

\paragraph{Building a canonical global restriction.}
Given selected disjoint blocks $(i_k,T_k)$, we choose one representative active $B$-index $j_k\in T_k$ for each $k$, and set $I=\{i_1,\dots,i_r\}$ and $J=\{j_1,\dots,j_r\}$.
We then define a single global restriction by setting $u_{A\setminus I}\leftarrow 0$ and, on $B$, setting all variables in $\bigcup_k(T_k\setminus\{j_k\})$ to $1$ while setting all other non-active $B$-variables to $0$, exactly as in \eqref{eq:alpha-bridge-results}.
This is the canonical choice that linearizes each block into $u_{i_k}v_{j_k}$ if no other cross residuals interfere.

\paragraph{Verification phase: residual-freeness test and core extraction.}
Verification is performed at the level of hyperedges, keeping careful track of $\mathbb{F}_2$ cancellation.
For each cross hyperedge $e\in E_{AB}$ we split $e$ into $e_A=e\cap A$ and $e_B=e\cap B$, apply the restriction $(u_{A\setminus I}=\beta,\ v_{B\setminus J}=\alpha)$, and compute the surviving active support
\begin{equation}
e' \;\coloneqq\; (e_A\cap I)\ \cup\ (e_B\cap J)\ \subseteq I\cup J,
\label{eq:reduced-support}
\end{equation}
discarding $e$ entirely if any fixed-to-$0$ variable lies in $e$.
Because the ANF is defined modulo $2$, we maintain a parity count of each resulting support $e'$ (equivalently, we toggle its presence in a hash map); only supports occurring an odd number of times survive.

Consider $\mathcal{M}$ as the set of surviving supports with odd parity. The residual-freeness condition holds if and only if there is no surviving support $e'\in\mathcal{M}$ that intersects both $I$ and $J$ and has size $|e'|\ge 3$. If this does not hold, higher-degree cross monomial terms will be present on the active variables and the candidate will not certify a bilinear core. This is because the bilinear cross monomials are exactly those $e'\in\mathcal{M}$ with $|e'|=2$ and with one element in $I$ and one in $J$.
These define $\Gamma_{\mathrm{core}}\in\mathbb{F}_2^{r\times r}$ by setting $(\Gamma_{\mathrm{core}})_{k\ell}=1$ precisely when $\{i_k,j_\ell\}\in\mathcal{M}$.
Any surviving monomials contained entirely in $I$ or entirely in $J$ contribute only to the local terms $\phi$ and $\psi$ in \eqref{eq:core-decomposition-results}.
Finally, we compute $t=\mathrm{rank}_{\mathbb{F}_2}(\Gamma_{\mathrm{core}})$ by Gaussian elimination over $\mathbb{F}_2$ and invoke Theorem~\ref{thm:core-certificate-results} to conclude $\mathrm{SR}_{A|B}\ge 2^t$.

\begin{remark}
\label{rem:complexity-results}
For a fixed candidate $(I,J;\beta,\alpha)$, verification is polynomial in the input size: it requires a single pass over $E_{AB}$ to compute the reduced supports and an $O(r^3)$ elimination to obtain $\mathrm{rank}_{\mathbb{F}_2}(\Gamma_{\mathrm{core}})$.
The combinatorial search for large disjoint block families can be hard in the worst case, but in the structured regimes that motivate this work (collections of near-disjoint bridges, sparse cross incidence, or planted core patterns) simple greedy selection with shallow local modifications often produces large certificates.
\end{remark}

\begin{example}[Vanishing bilinear slice but exponential certificate from disjoint CCZ bridges]
\label{ex:vanishing-bilinear-slice-exponential-core}
Let
\[
A=\{1,4,7\},
\qquad
B=\{2,3,5,6,8,9\},
\]
and take three cross hyperedges
\[
\{1,2,3\},\qquad \{4,5,6\},\qquad \{7,8,9\}.
\]
Then the cross-cut phase is
\[
f_{AB}(u,v)=u_1v_2v_3 \ \oplus\ u_4v_5v_6 \ \oplus\ u_7v_8v_9,
\]
which contains no degree-$2$ cross monomials, so the global bilinear slice satisfies $\Gamma_{AB}=0$ and the baseline rule of Sec.~\ref{sec:baseline} is silent.

Choose active sets $I=\{1,4,7\}$ and $J=\{2,5,8\}$ and apply the restriction $v_3\leftarrow 1$, $v_6\leftarrow 1$, and $v_9\leftarrow 1$ (there are no non-active $A$ bits in this example).
On the active variables $(u_I,v_J)$, the restricted cross phase becomes
\[
\widetilde f^{(I,J;\beta,\alpha)}_{\mathrm{core}}(u_I,v_J)
=
u_1v_2 \ \oplus\ u_4v_5 \ \oplus\ u_7v_8,
\]
with no residual higher-degree cross terms and no additional local terms.
Thus $\Gamma_{\mathrm{core}}=I_3$ and $t=3$, so Theorem~\ref{thm:core-certificate-results} yields the exponential Schmidt-rank certificate
\[
\mathrm{SR}_{A|B}\ \ge\ 2^{3}=8,
\]
despite $\Gamma_{AB}=0$.
\end{example}

\section{Discussion and Outlook}
\label{sec:discussion}

Our analysis reframes bipartite entanglement of hypergraph states in terms of the phase-cleaned sign matrix $R_{u,v}=(-1)^{f_{AB}(u,v)}$ across a cut $A|B$.
For graph states the cross phase is purely bilinear, and the entanglement structure is rigid: the Schmidt rank is exactly $2^{\mathrm{rank}_{\mathbb{F}_2}(\Gamma_{AB})}$, where $\Gamma_{AB}$ is the cross adjacency block.
Hypergraph states depart from this regime in two coupled ways.
First, the global bilinear slice may vanish even when the state is highly entangled, because high-degree cross hyperedges contribute only higher-order monomials to $f_{AB}$.
Second, even when a nontrivial bilinear slice is present, additional higher-degree cross terms can interfere with it and reduce the \emph{real} rank of $R$, so that $\mathrm{rank}_{\mathbb{F}_2}(\Gamma_{AB})$ ceases to be a reliable predictor unless higher-degree residuals are controlled (cf.\ Appendix~\ref{app:rank3-counterexample}).

The restriction-based viewpoint of Sec.~\ref{sec:results} addresses both issues by turning entanglement certification into the constructive task of isolating a \emph{residual-free} bilinear ``core'' on an active slice.
Formally, a single restriction selects a submatrix of $R$; if the restricted cross phase becomes bilinear up to cut-local terms, then a Walsh-character orthogonality argument certifies an exponential rank lower bound.
This converts the entanglement problem into a combinatorial one: finding large families of cross hyperedges that can be simultaneously linearized by one global assignment while avoiding cross-coupled residuals.
The disjoint bridge matching condition provides a clean sufficient criterion in this direction, and it captures the common situation of many CCZ-like bridges across a cut.

Several directions appear especially promising.
On the combinatorial side, it would be valuable to characterize when large residual-free cores must exist, for instance in random or structured hypergraph ensembles, and to understand the gap between the best provable core size and the true Schmidt rank.
On the algorithmic side, the verification phase of our method is efficient, but searching for large certificates can be hard in general; identifying tractable regimes and principled heuristics, possibly guided by Boolean-derivative structure, is likely important for practical use.
On the physics side, the present work focuses on Schmidt rank and postselected extraction; extending the techniques to bound entanglement entropies (such as R\'enyi-2 \cite{Adesso2012Renyi2Entropy} via Walsh averages \cite{XiaoMassey1988Walsh}, or von Neumann entropy directly from hyperedge structure remains an interesting challenge, as does generalizing beyond $\{0,\pi\}$ phases to weighted hypergraph states with arbitrary diagonal phases.

\section{Conclusion}
\label{sec:conclusion}

We developed a combinatorial framework for certifying bipartite entanglement of quantum hypergraph states from their Boolean phase polynomials.
Across any cut $A|B$, phase cleaning reduces the problem to the real sign matrix $R_{u,v}=(-1)^{f_{AB}(u,v)}$, whose rank equals the Schmidt rank.
In the purely bilinear setting this recovers the standard graph-state cut-rank rule, $\mathrm{SR}_{A|B}=2^{\mathrm{rank}_{\mathbb{F}_2}(\Gamma_{AB})}$.
For general hypergraph states with higher-degree cross terms, we defined residual-free bilinear cores exposed by a single restriction of non-active variables, and derived an explicit exponential Schmidt-rank certificate in terms of the $\mathbb{F}_2$ rank of the exposed core matrix $\Gamma_{\mathrm{core}}$.
We gave a sufficient combinatorial condition based on disjoint bridge matching and described an efficient verification procedure that turns candidate restrictions into rigorous certificates.

Together, these results offer a practical route to understanding and certifying entanglement structure in hypergraph states beyond the reach of the global bilinear slice, and they open several avenues for improving average extractable entanglement and for extending the framework to broader classes of phase states.

% =========================================================
\appendix
\section{Proof of Theorem~\ref{thm:bilinear-cutrank}}
\label{app:th1-proof}
For completeness, we provide here a proof of Theorem~\ref{thm:bilinear-cutrank}. This can be found, for example in Ref.~\cite{Hein2006}.
\begin{proof}
By performing phase cleaning (Sec.~\ref{subsec:prelim-bipartition}) one may directly compute the Schmidt rank from the sign matrix $R$:
$\mathrm{SR}_{A|B}=\mathrm{rank}(M)=\mathrm{rank}(R)$, according to Eq.~\eqref{eq:M-factor} and Eq.\eqref{eq:SR-rankR-prelim}.
Under the assumption \eqref{eq:pure-bilinear-assumption}, the sign matrix elements are
\[
R_{u,v} = (-1)^{u^\top \Gamma_{AB} v}.
\]
For each $u\in\{0,1\}^a$, define the row function (character)
\[
\chi_u:\{0,1\}^b\to\{\pm1\},
\qquad
\chi_u(v)\coloneqq (-1)^{u^\top \Gamma_{AB} v}.
\]
Note that $\chi_u$ depends only on the linear form $u^\top\Gamma_{AB}$ (equivalently, on $\Gamma_{AB}^\top u\in\mathbb{F}_2^b$).
Thus two rows $\chi_u$ and $\chi_{u'}$ will coincide if and only if $u^\top\Gamma_{AB}={u'}^\top\Gamma_{AB}$ as functions of $v$ (that means if and only if $u\oplus u'\in\ker(\Gamma_{AB}^\top)$).
Because  $\mathrm{rank}_{\mathbb{F}_2}(\Gamma_{AB}^\top)=k$, the image of $\Gamma_{AB}^\top$ has size $2^k$, and thus there are exactly $2^k$ distinct rows.

Furthermore, distinct rows are orthogonal under the uniform inner product on $\{0,1\}^b$:
for any $u,u'$,
\begin{align*}
\sum_{v\in\{0,1\}^b} \chi_u(v)\chi_{u'}(v)
&=
\sum_{v} (-1)^{(u\oplus u')^\top \Gamma_{AB} v}\\
&=
\begin{cases}
2^b, & (u\oplus u')\in\ker(\Gamma_{AB}^\top),\\
0, & \text{otherwise.}
\end{cases}
\end{align*}
Therefore the $2^k$ distinct rows form a mutually orthogonal set in $\mathbb{R}^{2^b}$ and are linearly independent. They represent the rows of the $R$ matrix, hence $\mathrm{rank}(R)=2^k$.
This gives $\mathrm{SR}_{A|B}=2^k$ and $E_{\mathrm{SR}}(A|B)=k$, which is a well known result for bipartite entanglement in graph states \cite{Hein2006}.
\end{proof}
\section{Rank-drop counterexample: higher-degree cross terms can reduce real rank}
\label{app:rank3-counterexample}

This appendix gives an explicit four-qubit example demonstrating that higher-degree cross terms can \emph{reduce} the real rank of the phase-cleaned sign matrix $R_{u,v}=(-1)^{f_{AB}(u,v)}$, and hence reduce the Schmidt rank across a cut below the bilinear prediction $2^{\mathrm{rank}_{\mathbb{F}_2}(\Gamma_{AB})}$. This motivates the residual-freeness requirement in Definition~\ref{def:residualfree-core-results}.

Consider a cut with two qubits on each side,
\[
A=\{1,2\},\qquad B=\{3,4\},
\]
and write the $A$-bits as $u=(u_1,u_2)\in\{0,1\}^2$ and the $B$-bits as $v=(v_1,v_2)\in\{0,1\}^2$.
Let the cross-cut phase polynomial be
\begin{equation}
f_{AB}(u,v)
\;=\;
u_1v_1 \ \oplus\ u_2v_2 \ \oplus\ u_1u_2v_1.
\label{eq:rankdrop-f}
\end{equation}
The bilinear slice is the first two terms, hence
\begin{equation}
\Gamma_{AB}
=
\begin{pmatrix}
1 & 0\\
0 & 1
\end{pmatrix},
\qquad
\mathrm{rank}_{\mathbb{F}_2}(\Gamma_{AB})=2.
\label{eq:rankdrop-Gamma}
\end{equation}
If the cross phase were purely bilinear, Theorem~\ref{thm:bilinear-cutrank} would give $\mathrm{SR}_{A|B}=2^2=4$.
However, the cubic cross term $u_1u_2v_1$ survives on the active variables and changes the real sign matrix.

Define the phase-cleaned sign matrix $R\in\{\pm1\}^{4\times 4}$ by
\[
R_{u,v}=(-1)^{f_{AB}(u,v)},
\]
with rows indexed by $u\in\{00,01,10,11\}$ and columns indexed by $v\in\{00,01,10,11\}$ in lexicographic order (with $v=(v_1,v_2)$).
A direct evaluation yields
\begin{equation}
R \;=\;
\begin{pmatrix}
\phantom{-}1 & \phantom{-}1 & \phantom{-}1 & \phantom{-}1 \\
\phantom{-}1 & -1 & \phantom{-}1 & -1 \\
\phantom{-}1 & \phantom{-}1 & -1 & -1 \\
\phantom{-}1 & -1 & \phantom{-}1 & -1
\end{pmatrix}.
\label{eq:rankdrop-R}
\end{equation}

The rank upper bound is immediate: the rows corresponding to $u=01$ and $u=11$ coincide, so $\mathrm{rank}(R)\le 3$.
To see that the rank is not smaller, note that the first three rows of Eq.~\eqref{eq:rankdrop-R} are nonzero and mutually orthogonal in $\mathbb{R}^4$, hence linearly independent.
Equivalently, the $3\times 3$ minor on the first three rows and first three columns has determinant
\begin{equation}
\det\!\begin{pmatrix}
1 & 1 & 1\\
1 & -1 & 1\\
1 & 1 & -1
\end{pmatrix}
=4\neq 0,
\label{eq:rankdrop-det}
\end{equation}
so $\mathrm{rank}(R)\ge 3$.
Therefore $\mathrm{rank}(R)=3$, and since $\mathrm{SR}_{A|B}=\mathrm{rank}(R)$ for REW states after phase cleaning (Sec.~\ref{subsec:prelim-bipartition}), we obtain
\begin{equation}
\mathrm{SR}_{A|B}=3
\;<\;
4
=
2^{\mathrm{rank}_{\mathbb{F}_2}(\Gamma_{AB})}.
\label{eq:rankdrop-conclusion}
\end{equation}

It is instructive to see how the cubic term produces the drop.
For $u=11$, the phase Eq.~\eqref{eq:rankdrop-f} reduces to
\[
f_{AB}(11,v)=v_1\oplus v_2 \oplus v_1 = v_2,
\]
so the $u=11$ row equals $(-1)^{v_2}$, which is exactly the $u=01$ row (since $f_{AB}(01,v)=v_2$).
Thus the higher-degree cross term cancels the $v_1$-dependence that would be present in the purely bilinear case, merging two Walsh-character rows and reducing the real rank.
This example is precisely the failure mode excluded by the residual-freeness condition in Definition~\ref{def:residualfree-core-results}.

\section*{Acknowledgment}
During the preparation of this manuscript, the authors used \textit{Generative AI (ChatGPT 5.2)} to improve the spelling, grammar, and readability of the text. After using this tool, the authors reviewed and edited the content as needed and take full responsibility for the final version of the manuscript.

\bibliographystyle{apsrev4-2}
\bibliography{bibliography}% Produces the bibliography via BibTeX.

\end{document}